\include{BoxedEPS}

\documentclass[11pt]{amsart}
\usepackage{graphicx, amsmath,amsthm}

\def\bi{\begin{itemize}}
\def\ei{\end{itemize}}
\def\bc{\begin{center}}
\def\ec{\end{center}}

\def\ba{ \begin{array}}
\def\ea{\end{array}}

\newcommand\be{\begin{equation}}
\newcommand\ee{\end{equation}}

\def\bc{\begin{center}}
\def\ec{\end{center}}

\newcommand\RR{{\mathbb R}}
\newcommand\ZZ{{\mathbb Z}}

\newtheorem{thm}{Theorem}[section]
\newtheorem{cor}[thm]{Corollary}
\newtheorem{lem}[thm]{Lemma}
\newtheorem{prop}[thm]{Proposition}

\numberwithin{equation}{section}

\begin{document}

\bibliographystyle{plain}

\title[Recovery of sparsest signals via $\ell^q$-minimization]
{Recovery of sparsest signals via
 $\ell^q$-minimization
}


\author{Qiyu Sun}
\address{Q. Sun, Department of Mathematics,  University of Central Florida,
Orlando, FL 32816, USA}

\email{qsun@mail.ucf.edu }


\date{\today }

%


\maketitle
\begin{abstract}
In this paper,  it is proved that every  $s$-sparse vector ${\bf x}\in {\mathbb R}^n$
 can be exactly recovered from
 the measurement vector ${\bf z}={\bf A} {\bf x}\in {\mathbb R}^m$ via  some
  $\ell^q$-minimization  with $0< q\le 1$, as soon as
  each $s$-sparse vector ${\bf x}\in {\mathbb R}^n$ is uniquely determined by the  measurement ${\bf z}$.
\end{abstract}

\section{Introduction and Main Results}

Define the norm $\|{\bf x}\|_q, 0\le q\le \infty$, of a vector  ${\bf x}=(x_1, \ldots, x_n)^T\in {\mathbb R}^n$
 by the number of its  nonzero components when $q=0$,
 the  quantity $(|x_1|^q+\cdots+|x_n|^q)^{1/q}$ when $0<q<\infty$,
 and  the maximum absolute value 
 $\max(|x_1|, \ldots, |x_n|)$ of its components when $q=\infty$.
We say that a vector ${\bf x}\in {\mathbb R}^n$ is  {\em $s$-sparse} if
$\|{\bf x}\|_0\le s$, i.e.,
the number of its nonzero components   is less than or equal to $s$.

\bigskip

In this paper, we consider the  problem  of compressive sensing  in  finding $s$-sparse  solutions
  ${\bf x}\in {\mathbb R}^n$ to the linear system
 \begin{equation}\label{section1.eq1}
{\bf A} {\bf x}= {\bf z}
 \end{equation}
 via solving the
$\ell^q$-minimization problem:
\begin{equation}\label{section1.eq2}
\min \|{\bf y}\|_q \quad {\rm subject\ to} \ {\bf A}{\bf y}={\bf z} 
\end{equation}
where $0<q\le 1$, $2\le 2s\le m\le n$, ${\bf A}$ is an
 $m\times n$ matrix, and    ${\bf z}\in {\mathbb R}^m$ is the observation data (\cite{blu08, candesrombergtao06, candestao05, char07, cdd09, Daubechies10}).

\bigskip

One of the  basic questions about finding $s$-sparse solutions to the linear system \eqref{section1.eq1} is
under what circumstances the linear system \eqref{section1.eq1} has a unique solution in $\Sigma_s$, the set of all $s$-sparse vectors.

 \begin{prop} {\rm (\cite{cdd09,donoho03})}\
 Let $2s\le m\le n$ and ${\bf A}$ be an
 $m\times n$ matrix. Then the following statements  are equivalent:
 \begin{itemize}

 \item[{(i)}] The  measurement ${\bf A}{\bf x}$ uniquely determines each $s$-sparse vector ${\bf x}$.

 \item[{(ii)}] There is a decoder $\Delta: {\mathbb R}^m\longmapsto {\mathbb R}^n$ such that $\Delta({\bf A}{\bf x})={\bf x}$ for all ${\bf x}\in \Sigma_s$.
 \item[{(iii)}] The only $2s$-sparse vector ${\bf y}$ that satisfies
   ${\bf A}{\bf y}={\bf 0}$  is the zero vector.

 \item[{(iv)}] There exist positive constants $\alpha_{2s}$ and $\beta_{2s}$ such that
 \begin{equation}\label{section1.eq3}
 \alpha_{2s}\|{\bf x}\|_2\le \|{\bf A}{\bf x}\|_2\le \beta_{2s}\|{\bf x}\|_2 \quad \ {\rm for \ all} \ {\bf x}\in \Sigma_{2s}.
 \end{equation}


  \end{itemize}
 \end{prop}

\bigskip

 The first contribution of this paper is to provide another equivalent statement:
 \begin{itemize}
 \item[{(v)}] {\em  There exists $0<q\le 1$ such that the decoder
$\Delta: {\mathbb R}^m\longmapsto {\mathbb R}^n$ defined by
\begin{equation} \Delta({\bf z}):={\rm argmin}_{{\bf Ay}={\bf z}} \|{\bf y}\|_q \end{equation}
 satisfies
 $\Delta({\bf A}{\bf x})={\bf x}$ for all ${\bf x}\in \Sigma_s$.
 }
 \end{itemize}

 The implication from (v) to (ii) is obvious. Hence it suffices to prove the implication from (iv) to  (v). For this,
 we  recall the {\em restricted isometry property} of order $s$ for an $m\times n$ matrix ${\bf A}$, i.e., there exists a positive constant
  $\delta\in (0,1)$ such that
  \begin{equation}\label{section1.eq4}
 (1- \delta )\|{\bf x}\|_2^2\le \|{\bf A}{\bf x}\|_2^2\le (1+\delta)\|{\bf x}\|_2^2\quad \ {\rm for \ all} \  {\bf x}\in \Sigma_s.
  \end{equation}
  The smallest positive constant $\delta$  that satisfies \eqref{section1.eq4}, to be denoted by $\delta_s({\bf A})$, is  known as the {\em restricted isometry constant} \cite{candesrombergtao06, candestao05}. Notice that  given a matrix ${\bf A}$ that satisfies \eqref{section1.eq3}, its rescaled matrix ${\bf B}:=
 \sqrt{2/(\alpha_{2s}^2+\beta_{2s}^2)} {\bf A}$  has the restricted isometry property  of order $2s$
  and its restricted isometry constant is given by $(\beta_{2s}^2-\alpha_{2s}^2)/(\alpha_{2s}^2+\beta_{2s}^2)$.
  Therefore the implication from (iv) to (v) further reduces to  establishing the following result:

  \begin{thm} \label{maintheorem1} Let integers $m,n$ and $s$ satisfy $2s\le m\le n$.
 If ${\bf A}$ is an
 $m\times n$ matrix  with $\delta_{2s}({\bf A})\in (0,1)$, then  there exists $0<q\le 1$
 such that any $s$-sparse vector ${\bf x}$ can be exactly recovered by solving the $\ell^q$-minimization problem:
 \begin{equation}\label{section1.eq5}
  \min \|{\bf y}\|_q \quad {\rm subject \ to } \ \ {\bf A}{\bf y}={\bf A}{\bf x}.
 \end{equation}
  \end{thm}

 The  above existence theorem about $\ell^q$-minimization  is established in \cite{foucartlai09} and
 \cite{char07}
 under a stronger assumption that $\delta_{2s+2}({\bf A})\in (0,1)$ and $\delta_{2s+1}({\bf A})\in (0,1)$
 respectively,  as it is obvious that
  $\delta_{2s}({\bf A})\le \delta_{2s+1}({\bf A})\le \delta_{2s+2} ({\bf A})$ for any $m\times n$ matrix ${\bf A}$.

 \bigskip

Given integers $s, m$ and $n$ satisfying $2s\le m\le n$ and  an $m\times n$ matrix ${\bf A}$, define
\begin{eqnarray}\label{section1.eq6-}
q_{s}({\bf A}) & := & \sup \big \{q\in [0,1]\big|\
{\rm any\  vector} \  {\bf x}\in \Sigma_s\  {\rm
can\ be\  exactly\ recovered}\nonumber\\
 & & {\rm  by \ solving\ the} \ \ell^q-{\rm minimization\ problem}\ \eqref{section1.eq5} \big\}.
\end{eqnarray}
Then $q_{s}({\bf A})>0$ whenever $\delta_{2s}({\bf A})<1$ by Theorem \ref{maintheorem1}.
It is also known that 
  any $s$-sparse vector ${\bf x}\in {\mathbb R}^n$
 can be exactly recovered by solving the $\ell^q$-minimization problem  \eqref{section1.eq5}
whenever  $q<q_{s}({\bf A})$
 \cite{gribonvalnielson07}. This establishes the equivalence among different $q\in [0, q_s({\bf A}))$
  in recovering  $s$-sparse solutions via
 solving the $\ell^q$-minimization problem \eqref{section1.eq5}.
Hence in order to recover sparsest vector ${\bf x}$ from the measurement ${\bf A}{\bf x}$,
 one may solve the $\ell^q$-minimization problem \eqref{section1.eq5} for some $0<q\le 1$ rather than
the $\ell^0$-minimization problem. Empirical evidence (\cite{char07, saab08,saab10}) strongly indicates that solving the $\ell^q$-minimization problem with $0<q\le 1$ takes much less time than with $q=0$.

The $\ell^0$-minimization problem is a combinatorial optimization  problem and NP-hard to solve \cite{nat95}, while on the other hand the $\ell^1$-minimization is convex and polynomial-time solvable
\cite{boyd04}. To guarantee the equivalence between
the $\ell^0$ and $\ell^1$-minimization problems \eqref{section1.eq5} in finding the  sparse
vector ${\bf x}$ from its measurement ${\bf A}{\bf x}$,
one needs to meet various requirements
on the matrix ${\bf A}$,  for instance, $\delta_s({\bf A})+\delta_{2s}({\bf A})+\delta_{3s}({\bf A})<1$ in \cite{candes06}, $\delta_{3s}({\bf A})+3\delta_{4s}({\bf A})<2$ in \cite{candesrombergtao06}, and
$\delta_{2s}({\bf A})<1/3\approx 0.3333, \sqrt{2}-1\approx 0.4142, 2/(3+\sqrt{2})\approx 0.4531,  2/(2+\sqrt{5})\approx 0.4731, 4/(6+\sqrt{6})\approx 0.4734$ in \cite{cdd09, candes, foucartlai09, cai10, foucart10} respectively.
Many random matrices with i.d.d. entries satisfy those requirement  to guarantee the equivalence \cite{candestao05},  but lots of deterministic matrices do not. In particular,
  matrices ${\bf A}_\epsilon$ are constructed in
 \cite{davisgribonval09} for any $\epsilon>0$ such that  $\delta_{2s}({\bf A}_{\epsilon})<1/\sqrt{2}+\epsilon$
 and that it fails on the recovery of  some $s$-sparse vectors ${\bf x}$ by
solving the $\ell^1$-minimization problem \eqref{section1.eq5} with ${\bf A}$ replaced by ${\bf A}_\epsilon$.

The $\ell^q$-minimization problem \eqref{section1.eq5}
with $0<q<1$ is more difficult to solve than the $\ell^1$-minimization problem
due to the nonconvexity and nonsmoothness. In fact, it is NP-hard to find a global minimizer in general
 but polynomial-time doable to find local minimizer \cite{jiang10}.
 Various  algorithms have been developed to solve the $\ell^q$-minimization problem \eqref{section1.eq5}, see for instance \cite{candeswakin08, chen09, Daubechies10, foucartlai09,rao99}.

\bigskip

For any $\delta\in (0,1)$,   define
\begin{eqnarray}\label{section1.eq6}
q_{\max}(\delta; m,n,s)  :=  \inf_{\delta_{2s}({\bf A})\le \delta} q_s({\bf A}).
\end{eqnarray}
Then
given  any positive number $q<q_{\max}(\delta; m,n,s)$ and any $m\times n$ matrix
${\bf A}$  with $\delta_{2s}({\bf A})\le \delta$,
 any  vector ${\bf x}\in \Sigma_s$ can be  exactly recovered  by  solving the $\ell^q$-minimization problem
 \eqref{section1.eq5}.
 For any $0<q\le 1$ and sufficiently small $\epsilon$, matrices ${\bf A}_{\epsilon}$ of  size $(n-1)\times n$
 are constructed in \cite{davisgribonval09}
   such that $\delta_{2s}({\bf A}_{q,\epsilon})<\frac{\eta_q}{2-q-\eta_q}+\epsilon$
and there is an $s$-sparse vector which cannot be recovered exactly by solving
  the $\ell^q$-minimization problem \eqref{section1.eq5} with ${\bf A}$  replaced by ${\bf A}_{q,\epsilon}$,   where
  $\eta_q$ is the unique positive solution to $\eta_q^{2/q}+1=2(1-\eta_q)/q$.
  The above construction of matrices for which the $\ell^q$-minimization fails to recover $s$-sparse vectors, together with
  the asymptotic estimate $\eta_q= 1-qx_0+o(q)$ as $q\to 0$, gives that
    \begin{equation}\label{section1.eq7}
  \limsup_{\delta\to 1-} \frac{q_{\max}(\delta; n-1, n, s)}{1-\delta}\le
  \lim_{q\to 0+} \frac{q(2-q-\eta_q)}{2-q-2\eta_q} =\frac{1}{2x_0-1}\approx 3.5911,
  \end{equation}
 where $x_0$ is the unique positive solution of the equation
  $e^{-2x}=2x-1$.
  The second contribution of this paper is a lower bound estimate for $q_{\max}(\delta; m,n,s)$
   as $\delta\to 1-$.

  \begin{thm}\label{maintheorem2} Let $q_{\max}(\delta; m,n,s)$ be defined as in \eqref{section1.eq6}. Then
  \begin{equation}\label{section1.eq8}
  \liminf_{\delta\to 1-} \frac{q_{\max}(\delta; m,n,s)}{1-\delta}\ge   \frac{e}{4} \approx 0.6796.  \end{equation}
  \end{thm}

\bigskip

Denote by ${\bf v}_S$ the vector
 which equals to ${\bf v}\in \RR^n$ on $S$ and vanishes on the complement $S^c$ where $S\subset\{1, \ldots, n\}$.
 We say that an $m\times n$ matrix ${\bf A}$ has the {\em null space property}
  of order $s$ in $\ell^q$ if there exists a positive constant $\gamma$ such that
  \begin{equation}\label{maintheorem4.eq1}
  \|{\bf h}_S\|_q\le \gamma \|{\bf h}_{S^c}\|_q
  \end{equation}
 hold for all ${\bf h}$ satisfying ${\bf A}{\bf h}={\bf 0}$ and all sets $S$ with its cardinality $\# S$  less than or equal to $s$
 (\cite{cdd09}).  The minimal constant $\gamma$ in \eqref{maintheorem4.eq1} is known as the {\em null space constant}.

  For  $0<q \le 1$ and $\delta\in (0,1)$,  define
\begin{eqnarray}\label{section1.eq9}
a(q, \delta) & := & \inf_{0<r_0<1} \max\Big\{\frac{1+r_0\delta}{(1+r_0^q\delta^q)^{1/q}},  \sup_{\sqrt{2}(1-r_0)\delta/2\le y\le 1}
  \frac{2 y}
 {\big(1+ 2^{-q/2}  y^{2+q}\big)^{1/q}}, \nonumber\\
& & \quad   \sup_{\sqrt{2}(1-r_0)\delta/2\le y\le 1}
 \frac{3y}
 {\big(1+y\big)^{1/q}}, \  \sup_{  1\le y}
 \frac{2y}
 {\big(1+y\big)^{1/q}}\Big\}.\end{eqnarray}
  The third contribution of this paper is the following result about the null space property of an  $m\times n$ matrix.

 \begin{thm}\label{maintheorem}
  Let $q$ be a positive number in $(0,1]$, integers $m, n$ and $s$ satisfy $2s\le m\le n$,
${\bf A}$ be an $m\times n$ matrix  with  $\delta_{2s}({\bf A})\in (0,1)$, and
set \begin{equation}\label{maintheorem3.eq1}
\delta_1:=\Big(\frac{1-\delta_{2s}({\bf A})}{1+\delta_{2s}({\bf A})}\Big)^{1/2}.\end{equation}
Then ${\bf A}$ has the null space property of order $s$  in $\ell^q$, and
its null space constant  is less than or equal to $a(q, \delta_1)/\delta_1$.
 \end{thm}

The fourth contribution of this paper is to show that one can stably reconstruct a  compressive signal
from noisy observation under the hypothesis that
\begin{equation}\label{maintheorem3.eq2} a(  q, \delta_1)<\delta_1.
\end{equation}

\begin{thm}\label{maintheorem3} Let $m, n$ and $s$ be integers  with $2s\le m\le n$,
${\bf A}$ be an $m\times n$ matrix  with  $\delta_{2s}({\bf A})\in (0,1)$, $\epsilon\ge 0$,
  $q\in (0,1]$ satisfy \eqref{maintheorem3.eq2} with $\delta_1$  given
 in \eqref{maintheorem3.eq1}, and
 ${\bf x}^*$ be the solution of the  $\ell^q$-minimization problem:
\begin{equation}\label{maintheorem3.neweq1}
\min_{\tilde {\bf x}\in {\mathbb R}^n} \|\tilde {\bf x}\|_q \ \  {\rm subject \ to} \  \ \|{\bf A}\tilde {\bf x}-{\bf y}\|_2\le \epsilon
\end{equation}
where ${\bf y}={\bf A}{\bf x}+{\bf z}$ is the observation corrupted with unknown noise ${\bf z}$,
 $\|{\bf z}\|_2\le \epsilon$
and
${\bf x}$ is the object we wish to reconstruct.
Then
\begin{equation}\label{maintheorem3.neweq2}
\|{\bf x}^*-{\bf x}\|_2\le C_0  s^{1/2-1/q} \|{\bf x}-{\bf x}_s\|_q+ C_1 \epsilon,
\end{equation}
and
\begin{equation}\label{maintheorem3.neweq3}
\|{\bf x}^*-{\bf x}\|_q\le C_2   \|{\bf x}-{\bf x}_s\|_q+ C_3 s^{1/q-1/2} \epsilon,
\end{equation}
where ${\bf x}_{s}$ be the best $s$-sparse vector in ${\mathbb R}^n$ to approximate  ${\bf x}_0$, i.e.,
$$
\|{\bf x}_{s}-{\bf x}\|_q=\inf_{{\bf x}'\in \Sigma_s} \|{\bf x}'-{\bf x}\|_q$$
and
$C_i, 0\le i\le 3$, are positive constants independent on $\epsilon, {\bf x}$ and $s$.
 \end{thm}

The  
 stable reconstruction of a compressive signal from its noisy observation
is established under various assumptions on the restricted isometry constant, for instance,
 $\delta_{3s}({\bf A})+3\delta_{4s}({\bf A})<2$  and $q=1$ in \cite{candesrombergtao06}, and
$\delta_{2s}({\bf A})< \sqrt{2}-1$ and $q=1$ in \cite{candes},
$\delta_{2t}({\bf A})< 2(\sqrt{2}-1) (t/s)^{1/q-1/2}/(1+2(\sqrt{2}-1) (t/s)^{1/q-1/2})$
for some $t\ge s$ and $0<q\le 1$ in \cite{foucartlai09}, and
$\delta_{ks}({\bf A})+k^{2/p-1}\delta_{(k+1)s}({\bf A})< k^{2/q}-1$ for some $k\in \ZZ/s$ and $0<q\le 1$
in \cite{saab08, saab10}.

 \bigskip
As an application of Theorem \ref{maintheorem3},
any $s$-sparse vector ${\bf x}$ can be exactly recovered by solving the $\ell^q$-minimization problem
 \eqref{section1.eq5} when  $q\in (0,1]$ satisfies \eqref{maintheorem3.eq2}.

\begin{cor}\label{maincorollary} Let $m, n$ and $s$ be integers  with $2s\le m\le n$,
${\bf A}$ be an $m\times n$ matrix  with  $\delta_{2s}({\bf A})\in (0,1)$, $\epsilon\ge 0$,
  $q\in (0,1]$ satisfy \eqref{maintheorem3.eq2} with $\delta_1$  given
 in \eqref{maintheorem3.eq1}. Then
 any $s$-sparse vector ${\bf x}$ can be exactly recovered by solving the $\ell^q$-minimization problem
 \eqref{section1.eq5}.
 \end{cor}

 Let
 $$q_{\rm succ}(\delta)={\tilde q}_{\max}\big(\sqrt{(1-\delta)/(1+\delta)}\big)$$
where $ \tilde q_{\max}(\delta_1)=\sup\{q\in (0,1]| \ a(q,  \delta_1)<\delta_1\}$,
and let $q_{\rm fail}(\delta)$ be the solution of the equation
$$\Big(\frac{(2-q)\delta}{1+\delta}\Big)^{2/q}+1=\frac{2-2\delta+2q\delta}{q+q\delta}$$
if it exists and be equal to one otherwise.
Then by Theorem \ref{maintheorem3}, any $s$-sparse vector ${\bf x}$ can be exactly recovered by
solving the $\ell^q$-minimization problem \eqref{section1.eq5}
 when $q<q_{\rm succ}(\delta_{2s}({\bf A}))$, while by \cite{davisgribonval09}
there exists a matrix ${\bf A}$ with $\delta_{2s}({\bf A})\le \delta$ and an $s$-sparse vector
${\bf x}$  such that the vector ${\bf x}$ cannot be exactly recovered by
solving the $\ell^q$-minimization problem \eqref{section1.eq5}
 when $q>q_{\rm fail}(\delta)$.
  The functions  ${q}_{\rm succ}(\delta)$ and $q_{\rm fail}(\delta)$ are plotted in Figure
 \ref{figure1}.

\begin{figure}[hbt]
\centering
\begin{tabular}{c}
  \includegraphics[width=90mm]{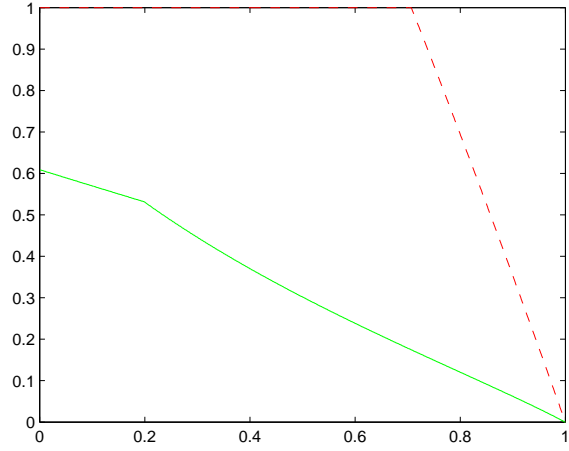}
\end{tabular}
\caption{The function ${q}_{\rm succ}(\delta)$ is plotted in continuous line, while the function
 $q_{\rm fail}(\delta)$ is plotted in
  dashed
line}
 \label{figure1}
\end{figure}

  \section{Proofs}

In this section, we give the proofs of Theorems \ref{maintheorem1}, \ref{maintheorem2}, \ref{maintheorem} and \ref{maintheorem3}.

\subsection{Proof of Theorem \ref{maintheorem}}

 To prove Theorem \ref{maintheorem}, we need three technical  lemmas.

\begin{lem} \label{inequality.lem}
Let $0<q\le 1, 0\le c\le 1$ and $a,b> 0$. Then
\begin{equation} \label{inequality.lem.eq1}
a+\sum_{k=1}^m t_k\le
\max\Big\{\max_{1\le k\le m} \frac{k+a}{(k+b)^{1/q}},  \frac{a+c}{(b+c^q)^{1/q}}\Big\}
\Big(b+\sum_{k=1}^m t_k^q\Big)^{1/q}
\end{equation}
holds for any
$(t_1,\ldots, t_m)\in [0,1]^m$ with
$t_1+\cdots+t_m\ge c$.
\end{lem}

\begin{proof}
Define
\begin{equation}\label{inequality.lem.pf.eq1}
F_{q,a,b,c}(m,n)=\sup_{\small \begin{array} {l} (t_1,\ldots, t_m)\in [0,1]^m\\
t_1+\cdots+t_m\ge c\end{array}} \frac{n+a+\sum_{k=1}^m t_k}{(n+b+\sum_{k=1}^m t_k^q)^{1/q}}.\end{equation}
By the method of  Lagrange multiplier, the function $(n+a+\sum_{k=1}^m t_k)(n+b+\sum_{k=1}^m t_k^q)^{-1/q}$
attains its maximum on the  boundary  or  on those points $(t_1, \ldots, t_m)$ whose components are  the same,
i.e.,
\begin{eqnarray*}
F_{q,a,b,c }(m,n) & = & \max\Big\{F_{q,a,b,0}(m-1, n+1), F_{q,a,b,c}(m-1, n), \nonumber\\
& & \qquad \qquad \sup_{c/m\le t\le 1} \frac{n+a+mt}{(n+b+mt^q)^{1/q}}\Big\}.
\end{eqnarray*}
As  the function $(n+a+mt)(n+b+mt^q)^{-1/q}$ has at most  one critical point  
and the second derivative at that critical point (if it exists) is positive, we then have
\begin{eqnarray} \label{inequality.lem.pf.eq2}
F_{q,a,b,c }(m,n)
& = & \max\Big\{F_{q,a, b, 0}(m-1, n+1), F_{q,a, b, c}(m-1, n),\nonumber\\
 & & \qquad \qquad  \frac{n+m+a}{(n+m+b)^{1/q}}, \frac{n+a+c}{(n+b+m^{1-q}c^q)^{1/q}}\Big\}.
 \end{eqnarray}
 Applying \eqref{inequality.lem.pf.eq2} iteratively we obtain
 \begin{eqnarray}\label{inequality.lem.pf.eq3}
F_{q,a,b,c }(m,n) & = & \max\Big\{F_{q,a, b, 0}(m-2, n+2), F_{q,a, b, 0}(m-2, n+1),\nonumber\\
 & & \qquad  F_{q,a, b, c}(m-2, n),  \frac{n+1+a}{(n+1+b)^{1/q}}, \frac{n+m-1+a}{(n+m-1+b)^{1/q}}, \nonumber\\
 & & \qquad \frac{n+m+a}{(n+m+b)^{1/q}}, \frac{n+a+c}{(n+b+(m-1)^{1-q}c^q)^{1/q}}
\Big\}\nonumber\\
  & = & \cdots\nonumber\\
& = &  \max\Big\{F_{q,a, b, 0}(1, n+m-1), \cdots, F_{q,a, b, 0}(1, n+1),\nonumber\\
 & & \qquad   F_{q,a, b, c}(1, n), \frac{n+a+c}{(n+b+2^{1-q}c^q)^{1/q}},\nonumber\\
 & & \qquad   \frac{n+m+a}{(n+m+b)^{1/q}}, \ldots, \frac{n+2+a}{(n+2+b)^{1/q}}, \frac{n+1+a}{(n+1+b)^{1/q}}\Big\}
 \nonumber\\
 & = & \max\Big\{\max_{1\le k\le m} \frac{n+k+a}{(n+k+b)^{1/q}},  \frac{n+a+c}{(n+b+c^q)^{1/q}}\Big\}.
\end{eqnarray}
Then   the conclusion \eqref{inequality.lem.eq1} follows by letting $n=0$ in the above estimate.
\end{proof}

\begin{lem}\label{inequality.lem2}
Let $0<q\le 1$, $c_1, c_2\in [0,1]$ and $a_i, b_i>0$ for $i=1,2, 3$. Then
\begin{eqnarray}\label{inequality.lem2.eq1}
& &  a_1+a_2 x+a_3 y +\sum_{k=1}^m t_k\nonumber\\
 & \le &
\max\Big\{\frac{a_1+a_2  } {(b_1+b_2)^{1/q}},
\frac{a_1+a_2c_1  } {(b_1+b_2c_1^q)^{1/q}},  \sup_{0\le l\le m}
 \frac{a_1+a_2+(a_3+l)c_2} {(b_1+b_2+(b_3+l)c_2^q)^{1/q}},\nonumber\\
 & &  \sup_{0\le l\le m}
 \frac{a_1+a_2 c_1+(a_3+l)c_2} {(b_1+b_2c_1^q+(b_3+l)c_2^q)^{1/q}}
 \Big\}  \times \Big(b_1+b_2 x^q+b_3 y^q +\sum_{k=1}^m t_k^q\Big)^{1/q}
\end{eqnarray}
holds for all $0\le t_1, \ldots, t_m\le y$, $c_1\le x\le 1$ and $0\le y\le c_2$.
\end{lem}

\begin{proof} Note that
the maximum values
 of the function $(a+bt)/(c+dt^q)^{1/q}$  on any closed subinterval of $[0,\infty)$
are attained on its boundary. Then
\begin{eqnarray*}
& & \frac{a_1+a_2 x+a_3 y +\sum_{k=1}^m t_k} {(b_1+b_2 x^q+b_3 y^q +\sum_{k=1}^m t_k^q)^{1/q}}\nonumber \\
 & = &
 \sup_{0\le l\le m}  \frac{a_1+a_2 x+(a_3+l) y } {(b_1+b_2 x^q+(b_3+l) y^q)^{1/q}}\\
 & = & \max\Big\{\frac{a_1+a_2 x } {(b_1+b_2 x^q)^{1/q}}, \sup_{0\le l\le m}
 \frac{a_1+(a_3+l)c_2+a_2 x} {(b_1+(b_3+l)c_2^q+b_2 x^q)^{1/q}} \Big\}\nonumber\\
 & \le & \max\Big\{\frac{a_1+a_2  } {(b_1+b_2)^{1/q}}, \frac{a_1+a_2c_1  } {(b_1+b_2c_1^q)^{1/q}}, \sup_{0\le l\le m}
 \frac{a_1+a_2+(a_3+l)c_2} {(b_1+b_2+(b_3+l)c_2^q)^{1/q}},\nonumber\\
 & &  \sup_{0\le l\le m}
 \frac{a_1+a_2 c_1+(a_3+l)c_2} {(b_1+b_2c_1^q+(b_3+l)c_2^q)^{1/q}}
 \Big\},
\end{eqnarray*}
and
\eqref{inequality.lem2.eq1} follows.
\end{proof}

\begin{lem}\label{mainlemma}
Let $0<q\le 1$,  $s\ge 1$ be a positive  integer, and
let $\{a_j\}_{j\ge 1}$ be a finite decreasing sequence of nonnegative numbers with
\begin{equation}\label{mainlemma.eq1}
\sum_{k\ge 1} \Big(\sum_{i=1}^{s} a_{k s+i}^2\Big)^{1/2}\ge \delta \Big(\sum_{i=1}^s |a_{i}|^2\Big)^{1/2}\end{equation}
for some $\delta\in (0,1)$. Then
\begin{equation}\label{mainlemma.eq2}
\sum_{k\ge 1} \Big(\sum_{i=1}^{s} a_{k s+i}^2\Big)^{1/2}\le a(q, \delta) s^{1/2-1/q}\Big(\sum_{j\ge 1} a_j^q\Big)^{1/q},
\end{equation}
where $a(q, \delta)$ is defined as in \eqref{section1.eq9}.
\end{lem}

\begin{proof} Clearly the conclusion \eqref{mainlemma.eq2} holds when $a_{s+1}=0$ for in this case the left hand side of
\eqref{mainlemma.eq2} is equal to 0. So we may assume that $a_{s+1}\ne 0$ from now on.
Let $r_0$ be an arbitrarily number in $(0,1)$.  To establish  \eqref{mainlemma.eq2}, we consider two cases.
 \bigskip

{\em Case I}: $\sum_{k\ge 2}  a_{ks+1}\ge   r_0\delta  a_{s+1}$. 

In this case,
\begin{eqnarray} \label{mainlemma.pf.eq1}
 & &  \frac{\sum_{k\ge 1} \big(\sum_{i=1}^{s} a_{k s+i}^2\big)^{1/2}}{\big(\sum_{j\ge 1} a_j^q\big)^{1/q}}\le
\frac{ s^{1/2} \sum_{k\ge 1} a_{ks+1} } { s^{1/q} \big(  \sum_{k\ge 1} a_{ks+1}^q \big)^{1/q}}\nonumber\\
& = &  s^{1/2-1/q}
\frac{ 1+ \sum_{k\ge 2} a_{ks+1}/a_{s+1} } { \big( 1+ \sum_{k\ge 2} (a_{ks+1}/a_{s+1})^q \big)^{1/q}}\nonumber\\
%
& \le &  s^{1/2-1/q}
\max\Big\{\frac{1+r_0\delta}{(1+r_0^q\delta^q)^{1/q}}, \max_{k\ge 1} \frac{k+1}{(k+1)^{1/q}}\Big\}\nonumber\\
& = & s^{1/2-1/q} (1+r_0\delta)(1+r_0^q\delta^q)^{-1/q}, 
\end{eqnarray}
where  the first inequality  holds because $\{a_j\}_{j\ge 1}$ is a decreasing sequence of nonnegative numbers,
 the  second  inequality follows from  Lemma \ref{inequality.lem}, and the last  equality is true as
  $(1+t) (1+t^q)^{-1/q}$ is a decreasing function on $(0,1]$.

\bigskip

{\em Case II}:  $\sum_{k\ge 2} a_{ks+1} < r_0\delta  a_{s+1}$.

Let $s_0$ be the smallest integer in $[1, s]$ satisfying $a_{s+s_0+1}/a_{s+1}\le (s_0/s)^{1/2}$.
The existence and uniqueness of such an integer $s_0$ follow from the decreasing property of the sequence
$\{a_{s+s_0+1}/a_{s+1}\}_{s_0=1}^s$, the increasing property of the sequence $\{(s_0/s)^{1/2}\}_{s_0=1}^s$, and
$a_{s+s_0+1}/a_{s+1}\le (s_0/s)^{1/2}$ when $s_0=s$.
Then from the decreasing property of the sequence $\{a_j\}_{j\ge 1}$ and the definition of the integer $s_0$ it follows that
\begin{equation} \label{mainlemma.pf.eq2}
\frac{a_{s+s_0}}{a_{s+1}}\ge \big(\frac{s_0-1}{s}\big)^{1/2}
\end{equation}
and
\begin{eqnarray*}
 \sqrt{2} s_0^{1/2} a_{s+1} & \ge &  \big(s_0 a_{s+1}^2+(s-s_0) \frac{s_0}{s} a_{s+1}^2\big)^{1/2}   \ge     \Big(\sum_{i=1}^s a_{s+i}^2\Big)^{1/2}\nonumber\\
  & \ge &
   \delta \Big(\sum_{i=1}^s a_i^2\Big)^{1/2} - \sum_{k\ge 2} \Big(\sum_{i=1}^s a_{ks+i}^2\Big)^{1/2}\nonumber\\
   & \ge & \delta s^{1/2} a_{s+1}- s^{1/2} \sum_{k\ge 2} a_{ks+1}
  \ge      (1-r_0)\delta s^{1/2}  a_{s+1},
\end{eqnarray*}
which implies that
\begin{equation} \label{mainlemma.pf.eq3}
s_0\ge \frac{(1-r_0)^2 \delta^2}{2} s.
\end{equation}
Applying  the decreasing property of the sequence $\{a_j\}$ and using the inequality
 $(\theta a^2+(1-\theta) b^2)^{1/2}\le \theta^{1/2} a+ (1-\theta^{1/2}) b$
where $a\ge b\ge 0$ and $\theta\in [0,1]$, we obtain
\begin{eqnarray}\label{mainlemma.pf.eq5}
 & & s^{-1/2} \sum_{k\ge 1} \Big(\sum_{i=1}^{s} a_{k s+i}^2\Big)^{1/2}\nonumber\\
& \le &  s^{-1/2} \big ((s_0-1) a_{s+1}^2+ a_{s+s_0}^2+ (s-s_0) a_{s+s_0+1}^2\big)^{1/2}\nonumber\\
 & & + s^{-1/2} \sum_{k\ge 2}
\big(s_0 a_{ks+1}^2+(s-s_0) a_{ks+s_0+1}^2\big)^{1/2}\nonumber\\
& \le & \sqrt{\frac{s_0}{s}}
\Big( \frac{s_0-1}{s_0} a_{s+1}^2+ \frac{1}{s_0} a_{s+s_0}^2\Big)^{1/2}
+\Big(1-\sqrt{\frac{s_0}{s}}\Big) a_{s+s_0+1}\nonumber\\
& & + \sum_{k\ge 2}\Big( \sqrt\frac{s_0}{s}
 a_{ks+1}+\Big(1-\sqrt{\frac{s_0}{s}}\Big) a_{ks+s_0+1} \Big)\nonumber \\
&\le &  \sqrt{\frac{s_0-1}{s}}  a_{s+1}+ \frac{\sqrt{s_0}-\sqrt{s_0-1}}{\sqrt{s}} a_{s+s_0}+
 \sum_{k\ge 1}  a_{ks+s_0+1},
\end{eqnarray}
and
\begin{eqnarray}\label{mainlemma.pf.eq4}
\sum_{j\ge 1} a_j^q & \ge &  (s+1) a_{s+1}^q+ (s_0-1)a_{s+s_0}^q\nonumber\\
& & \quad +a_{s+s_0+1}^q+s\sum_{k\ge 2} a_{ks+s_0+1}^q.
\end{eqnarray}
Combining \eqref{mainlemma.pf.eq5} and \eqref{mainlemma.pf.eq4}, recalling
\eqref{mainlemma.pf.eq2} and the definition of the integer $s_0$, and applying Lemma \ref{inequality.lem2}
with $c_1=\sqrt{(s_0-1)/s}$ and $c_2=\sqrt{s_0/s}$,
we get
\begin{eqnarray}\label{mainlemma.pf.eq6}
 & &
   s^{1/q-1/2} \frac{\sum_{k\ge 1} \big(\sum_{i=1}^{s} a_{k s+i}^2\big)^{1/2}}{\big(\sum_{j\ge 1} a_j^q\big)^{1/q}} \nonumber\\
  & \le &    \frac{\sqrt{\frac{s_0-1}{s}}  a_{s+1}+ \frac{\sqrt{s_0}-\sqrt{s_0-1}}{\sqrt{s}} a_{s+s_0}+
a_{s+s_0+1}+ \sum_{k\ge 2}  a_{ks+s_0+1}}
 {\big((1+1/s) a_{s+1}^q+ (s_0-1)a_{s+s_0}^q/s+a_{s+s_0+1}^q/s+ \sum_{k\ge 2} a_{ks+s_0+1}^q\big)^{1/q}} \nonumber\\
 &\le  &
 \max\Big\{ \frac{\sqrt{\frac{s_0}{s}} }
 {\big(1+s_0/s \big)^{1/q}},
  \frac{\sqrt{\frac{s_0-1}{s}} \Big(1+   \frac{\sqrt{s_0}-\sqrt{s_0-1}}{\sqrt{s}}\Big)}
 {\big(1+1/s + ((s_0-1)/s)^{1+q/2}\big)^{1/q}},\nonumber\\
& & \sup_{l\ge 0} \frac{(l+2)\sqrt{\frac{s_0}{s}}}
 {\big(1+s_0/s+(l+1/s)\sqrt{\frac{s_0}{s}} \big)^{1/q}},\nonumber\\
 && \sup_{l\ge 0} \frac{(l+1)\sqrt{\frac{s_0}{s}}+\sqrt{\frac{s_0-1}{s}}\Big(1+ \frac{\sqrt{s_0}-\sqrt{s_0-1}}{\sqrt{s}}\Big)}
 {\big(\big((l+1/s)\sqrt{\frac{s_0}{s}}+ 1+1/s\big) + ((s_0-1)/s)^{1+q/2}\big)^{1/q}}\Big\}.
 \end{eqnarray}
Therefore
\begin{eqnarray} \label{mainlemma.pf.eq7}
 & &
 \frac{\sum_{k\ge 1} \big(\sum_{i=1}^{s} a_{k s+i}^2\big)^{1/2}}{\big(\sum_{j\ge 1} a_j^q\big)^{1/q}} \nonumber\\
 & \le &  s^{1/2-1/q} \max\Big\{ \frac{\sqrt{{s_0}/{s}} }
 {\big(1+s_0/s \big)^{1/q}},
  \frac{\sqrt{{s_0}/{s}} }
 {\big(1+ 2^{-q/2}  (s_0/s)^{1+q/2}\big)^{1/q}},\nonumber\\
& & \sup_{l\ge 0} \frac{(l+2)\sqrt{{s_0}/{s}}}
 {\big(1+s_0/s+l\sqrt{{s_0}/{s}} \big)^{1/q}},
 \sup_{l\ge 0} \frac{(l+2)\sqrt{{s_0}/{s}}}
 {\big(1+l\sqrt{{s_0}/{s}}+ 2^{-q/2} (s_0/s)^{1+q/2}\big)^{1/q}}\Big\}\nonumber\\
 & \le &  s^{1/2-1/q} \max\Big\{
  \frac{2 \sqrt{{s_0}/{s}} }
 {\big(1+ 2^{-q/2}  (s_0/s)^{1+q/2}\big)^{1/q}},\  \sup_{l\ge 1} \frac{(l+2)\sqrt{{s_0}/{s}}}
 {(1+l\sqrt{{s_0}/{s}}\big)^{1/q}}\Big\}\nonumber\\
 & \le &   s^{1/2-1/q} \max\Big\{ \sup_{\sqrt{2}(1-r_0)\delta/2\le y\le 1}
  \frac{2 y}
 {\big(1+ 2^{-q/2}  y^{2+q}\big)^{1/q}}, \nonumber\\
& & \quad   \sup_{\sqrt{2}(1-r_0)\delta/2\le y\le 1}
 \frac{3y}
 {\big(1+y\big)^{1/q}}, \  \sup_{  1\le y}
 \frac{2y}
 {\big(1+y\big)^{1/q}}\Big\},
 \end{eqnarray}
 where the third inequality is valid by \eqref{mainlemma.pf.eq3} and
 the first inequality follows from
 the following two inequalities:
\begin{equation}
\sqrt{\frac{t-1}{s}}\Big(1+\frac{\sqrt{t}-\sqrt{t-1}}{\sqrt{s}}\Big)\le \sqrt{\frac{t}{s}}
\end{equation}
and
\begin{equation}
\frac{1}{s}+ \Big(\frac{t-1}{s}\Big)^{1+q/2}\ge  \Big(\frac{1}{s}\Big)^{1+q/2} + \Big(\frac{t-1}{s}\Big)^{1+q/2}\ge 2^{-q/2} \Big(\frac{t}{s}\Big)^{1+q/2}, \quad 1\le t\le s.
\end{equation}

The conclusion \eqref{mainlemma.eq2} follows from \eqref{mainlemma.pf.eq1} and \eqref{mainlemma.pf.eq7}.
 \end{proof}

Now we give the proof of Theorem \ref{maintheorem}.

\begin{proof}
[Proof of Theorem \ref{maintheorem}]
Let ${\bf h}$ satisfy
\begin{equation}\label{maintheorem.pf.eq1}
{\bf Ah}={\bf 0}\end{equation}
 and let $S_0$ be a subset of $\{1, \ldots, n\}$ with cardinality $\# S_0$  less than or equal to $s$.
We partition $S_0^c\subset \{1, \ldots, n\}$ as $S_0^c=S_1\cup\cdots\cup  S_l$,
where $S_1$ is the set of indices of the $s$ largest components, in absolute value, of ${\bf h}$ in $S_0^c$, $S_2$ is the set
of indices of the next $s$  largest components, in absolute value, of ${\bf h}$
 in $(S_0\cup S_1)^c$, and so on.
Applying the parallelogram identity, we obtain from the restricted isometry property \eqref{section1.eq4} that
\begin{equation}\label{maintheorem.pf.eq3}
|\langle {\bf Au}, {\bf Av}\rangle|\le \delta_{2s}({\bf A})  \|{\bf u}\|_2 \|{\bf v}\|_2
\end{equation}
for all  $s$-sparse vectors ${\bf u},{\bf v}\in \Sigma_s$  whose supports have empty intersection
\cite{candestao05}.
Combining \eqref{maintheorem.pf.eq1} and \eqref{maintheorem.pf.eq3} and using the restricted isometry property \eqref{section1.eq4}
yield
\begin{eqnarray} \label{maintheorem.pf.eq4}
& & (1-\delta_{2s}({\bf A}) )  \Big(\|{\bf h}_{S_0}\|_2^2+ \|{\bf h}_{S_1}\|_2^2)\nonumber\\
& \le & \langle {\bf A}({\bf h}_{S_0}+{\bf h}_{S_1}), {\bf A} ({\bf h}_{S_0}+{\bf h}_{S_1})\rangle\nonumber\\
& \le & \big\langle {\bf A}\big(\sum_{i\ge 2} {\bf h}_{S_i}\big), {\bf A} \big(\sum_{j\ge 2} {\bf h}_{S_j}\big)\big\rangle\nonumber\\
& \le & \sum_{i,j\ge 2} \delta_{2s} ({\bf A}) \|{\bf h}_{S_i}\|_2 \|{\bf h}_{S_j}\|_2+\sum_{j\ge 2}\|{\bf h}_{S_j}\|_2^2
\nonumber\\
&= &  \delta_{2s}({\bf A}) \Big(\sum_{j\ge 2} \|{\bf h}_{S_j}\|_2\Big)^2+ \sum_{j\ge 2}\|{\bf h}_{S_j}\|_2^2\nonumber\\
& \le &  (1+\delta_{2s}({\bf A})) \Big(\sum_{j\ge 2} \|{\bf h}_{S_j}\|_2\Big)^2,
\end{eqnarray}
which implies that
\begin{equation} \label{maintheorem.pf.eq5}
\sum_{j\ge 2} \|{\bf h}_{S_j}\|_2\ge \Big(\frac{1-\delta_{2s}({\bf A})}{1+\delta_{2s}({\bf A})}\Big)^{1/2} \|{\bf h}_{S_1}\|_2.
\end{equation}
Applying   Lemma \ref{mainlemma} with $\delta_1=\big(\frac{1-\delta_{2s}({\bf A}) }{1+\delta_{2s}({\bf A}) }\big)^{1/2}$
gives
\begin{equation}  \label{maintheorem.pf.eq6}
\sum_{j\ge 2} \|{\bf h}_{S_j}\|_2\le a(q, \delta_1) s^{1/2-1/q} \|{\bf h}_{S_0^c}\|_q.
\end{equation}
Then  substituting the above estimate for $\sum_{j\ge 2} \|{\bf h}_{S_j}\|_2$ into the right hand side of the inequality
\eqref{maintheorem.pf.eq4} and recalling that ${\bf h}_{S_0}$ is an  $s$-sparse vector lead to
\begin{equation}
\|{\bf h}_{S_0}\|_q\le s^{1/q-1/2}\|{\bf h}_{S_0}\|_2\le
 s^{1/q-1/2} (\delta_1)^{-1} \sum_{j\ge 2} \|{\bf h}_{S_j}\|_2\le
 \frac{a(q, \delta_1) }{\delta_1} \|{\bf h}_{S_0^c}\|_q,
\end{equation}
the desired null space property.
\end{proof}

\subsection{Proof of Theorem \ref{maintheorem3}}
We follow the argument in \cite{candes, candesrombergtao06}.
Set ${\bf h}={\bf x}^*-{\bf x}$, and denote
 by $S_0$ the support of the vector ${\bf x}_s\in \Sigma_s$,
by $S_0^c$ the complement of the set $S_0$ in $\{1,\ldots, n\}$. Then
\begin{equation}\label{maintheorem3.pf.eq1}
\|{\bf Ah}\|_2=\|{\bf A}{\bf x}^*-{\bf A}{\bf x}\|_2\le \|{\bf A} {\bf x}^*-{\bf y}\|_2+
\|{\bf z}\|_2\le 2\epsilon
\end{equation}
and
\begin{equation}\label{maintheorem3.pf.eq2}
\|{\bf h}_{S_0^c}\|_q^q\le \|{\bf h}_{S_0}\|_q^q+2 \|{\bf x}-{\bf x}_s\|_q^q,
\end{equation}
since
\begin{eqnarray*}
\|{\bf x}_{s}\|_q^q+ \|{\bf x}_{_{S_0^c}}\|_q^q & = &
\|{\bf x}\|_q^q\ge \|{\bf x}^*\|_q^q=\|{\bf x}_s+{\bf h}_{_{S_0}}\|_q^q+\|{\bf x}_{_{S_0^c}}+{\bf h}_{_{S_0^c}}\|_q^q\\
&  \ge  &
\|{\bf x}_s\|_q^q-\|{\bf h}_{_{S_0}}\|_q^q+\|{\bf h}_{S_0^c}\|_q^q-\|{\bf x}_{_{S_0^c}}\|_q^q.\end{eqnarray*}
Similar to the argument used in the proof of Theorem \ref{maintheorem}, we partition $S_0^c\subset \{1, \ldots, n\}$ as $S_0^c=S_1\cup\cdots\cup  S_l$,
where $S_1$ is the set of indices of the $s$ largest absolute-value component of ${\bf h}$ in $S_0^c$, $S_2$ is the set
of indices of the next $s$ largest absolute-value components of ${\bf h}$ on $S_0^c$, and so on.
Then  it follows from \eqref{section1.eq4}, \eqref{maintheorem.pf.eq4} and \eqref{maintheorem3.pf.eq1} that
\begin{eqnarray} \label{maintheorem3.pf.eq3}
& & (1-\delta_{2s}({\bf A}) )  \big(\|{\bf h}_{S_0}\|_2^2+ \|{\bf h}_{S_1}\|_2^2\big)\nonumber\\
& \le & \langle {\bf A}({\bf h}_{S_0}+{\bf h}_{S_1}), {\bf A} ({\bf h}_{S_0}+{\bf h}_{S_1})\rangle\nonumber\\
& \le & \big\langle {\bf A}{\bf h}-{\bf A}\big(\sum_{i\ge 2} {\bf h}_{S_i}\big),{\bf A}{\bf h}- {\bf A} \big(\sum_{j\ge 2} {\bf h}_{S_j}\big)\big\rangle\nonumber\\
& \le &  \big(2\epsilon+ \sqrt{1+\delta_{2s}({\bf A})} \sum_{j\ge 2} \|{\bf h}_{S_j}\|_2\big)^2.
\end{eqnarray}
By the continuity of the function $a(q, \delta)$ about $\delta\in (0,1)$ and the assumption \eqref{maintheorem3.eq2}, there exists a positive  number $r$ such that
\begin{equation}\label{maintheorem3.pf.eq4}
a(q, \delta_1/(1+r))< \delta_1/(1+r).
\end{equation}

\bigskip

If $\sum_{j\ge 2} \|{\bf h}_{S_j}\|_2\le 2 \epsilon/ (r \sqrt{1+\delta_{2s}({\bf A})})$, then it follows from
\eqref{maintheorem3.pf.eq2}, \eqref{maintheorem3.pf.eq3} and the fact that ${\bf h}_{S_0}\in \Sigma_s$ that
\begin{eqnarray} \label{maintheorem3.pf.eq5}
\|{\bf x}^*-{\bf x}\|_2 & = &
\|{\bf h}\|_2\le
\big(\|{\bf h}_{S_0}\|_2^2+ \|{\bf h}_{S_1}\|_2^2)^{1/2} +\sum_{j\ge 2} \|{\bf h}_{S_j}\|_2\nonumber\\
& \le &   2 \Big(\frac{(1+r) }{ r\sqrt{1-\delta_{2s}({\bf A})}}+\frac{ 1}{r \sqrt{1+\delta_{2s}({\bf A})}}\Big)\epsilon,
\end{eqnarray}
and
\begin{eqnarray}\label{maintheorem3.pf.eq6}
\|{\bf x}^*-{\bf x}\|_q^q & \le  & \|{\bf h}_{S_0}\|_q^q+  \|{\bf h}_{S_0^c}\|_q^q
\le 2 \|{\bf h}_{S_0}\|_q^q
+2 \|{\bf x}-{\bf x}_s\|_q^q
\nonumber\\
 & \le &  2 s^{1-q/2} \|{\bf h}_{S_0}\|_2^q
+2 \|{\bf x}-{\bf x}_s\|_q^q\nonumber\\
 & \le &  2^{1+q} \frac{(1+r)^q }{ r^q(1-\delta_{2s}({\bf A}))^{q/2}} s^{1-q/2}\epsilon^q +2 \|{\bf x}-{\bf x}_s\|_q^q.
\end{eqnarray}

\bigskip

If $\sum_{j\ge 2} \|{\bf h}_{S_j}\|_2\ge 2 \epsilon/ (r \sqrt{1+\delta_{2s}({\bf A})})$, then
\begin{equation}\label{maintheorem3.pf.eq7}
\delta_1 \big(\|{\bf h}_{S_0}\|_2^2+ \|{\bf h}_{S_1}\|_2^2\big)^{1/2}\le  (1+r) \sum_{j\ge 2} \|{\bf h}_{S_j}\|_2
\end{equation}
by \eqref{maintheorem3.pf.eq3}, where we set $\delta_1=\sqrt{(1-\delta_{2s}({\bf A}))/(1+\delta_{2s}({\bf A}))}$.
Using \eqref{maintheorem3.pf.eq7}  and applying Lemma \ref{mainlemma} with $\delta=\delta_1/(1+r)$ give
\begin{equation}\label{maintheorem3.pf.eq8}
\sum_{j\ge 2} \|{\bf h}_{S_j}\|_2\le a(q, \delta_1/(1+r)) s^{1/2-1/q} \|{\bf h}_{S_0^c}\|_q.
\end{equation}
 Noting the fact that ${\bf h}_{S_0}\in \Sigma_s$ and   then applying \eqref{maintheorem3.pf.eq2}, \eqref{maintheorem3.pf.eq7} and \eqref{maintheorem3.pf.eq8} yield
\begin{eqnarray*}
\|{\bf h}_{S_0}\|_q^q & \le &  s^{1-q/2} \|{\bf h}_{S_0}\|_2^q\le
\Big(\frac{a(q, \delta_1/(1+r))}{\delta_1/(1+r)}\Big)^q \|{\bf h}_{S_0^c}\|_q^q\\
 & \le &
\Big(\frac{a(q, \delta_1/(1+r))}{\delta_1/(1+r)}\Big)^q \|{\bf h}_{S_0}\|_q^q+2
\Big(\frac{a(q, \delta_1/(1+r))}{\delta_1/(1+r)}\Big)^q \|{\bf x}-{\bf x}_s\|_q^q.
\end{eqnarray*}
This, together with \eqref{maintheorem3.pf.eq4}, leads to the  following crucial estimate:
\begin{equation} \label{maintheorem3.pf.eq9}
\|{\bf h}_{S_0}\|_q^q\le \frac{2 (a(q, \delta_1/(1+r)))^q}{(\delta_1/(1+r))^q-(a(q, \delta_1/(1+r)))^q}\|{\bf x}-{\bf x}_s\|_q^q.
\end{equation}
Combining \eqref{maintheorem3.pf.eq2},
\eqref{maintheorem3.pf.eq7}, \eqref{maintheorem3.pf.eq8} and \eqref{maintheorem3.pf.eq9}, we obtain
\begin{eqnarray} \label{maintheorem3.pf.eq10}
\|{\bf x}^*-{\bf x}\|_2 & \le  &  \big(\|{\bf h}_{S_0}\|_2^2+ \|{\bf h}_{S_1}\|_2^2)^{1/2} +\sum_{j\ge 2} \|{\bf h}_{S_j}\|_2 \nonumber\\
& \le & \frac{2^{1/q}(1+r+\delta_1)\big ( a(q, \delta_1/(1+r))\big)^2 }
{\delta_1 \big((\delta_1/(1+r))^q-(a(q, \delta_1/(1+r)))^q\big)^{1/q} }
 s^{1/2-1/q} \|{\bf x}-{\bf x}_s\|_q ,
\end{eqnarray}
 and
 \begin{eqnarray} \label{maintheorem3.pf.eq11}
 \|{\bf x}^*-{\bf x}\|_q^q & \le & 2 \|{\bf h}_{S_0}\|_q^q
+2 \|{\bf x}-{\bf x}_s\|_q^q\nonumber\\
 &\le &
\frac{2 (\delta_1/(1+r))^q+2(a(q, \delta_1/(1+r)))^q}{(\delta_1/(1+r))^q-(a(q, \delta_1/(1+r)))^q}\|{\bf x}-{\bf x}_s\|_q^q.
 \end{eqnarray}

 \bigskip
 The desired error estimates  \eqref{maintheorem3.neweq2} and \eqref{maintheorem3.neweq3} follow from \eqref{maintheorem3.pf.eq5},
 \eqref{maintheorem3.pf.eq6}, \eqref{maintheorem3.pf.eq10} and \eqref{maintheorem3.pf.eq11}.

\subsection{Proof of Theorem \ref{maintheorem1}} The conclusion in Theorem \ref{maintheorem1} follows from Corollary \ref{maincorollary}  and the observation that
\begin{equation}\label{maintheorem1.pf.eq1}
\lim_{q\to 0+} a(q, \delta)=0\end{equation} for any $\delta\in (0,1)$.

\subsection{Proof of Theorem \ref{maintheorem2}}
 Let
 \begin{equation}\label{maintheorem2.pf.eq1}
 \tilde q_{\max}(\delta_1)=\sup\{q\in (0,1]| \ a (q,  \delta_1)<\delta_1\}.
 \end{equation}
  Take sufficiently small $\epsilon>0$.
Note that  \begin{eqnarray}\label{maintheorem2.pf.eq2}
& &  \sup_{\sqrt{2}(1-r_0)\delta_1/2\le y\le 1}
  \frac{2 y}
 {\big(1+ 2^{-q/2}  y^{2+q}\big)^{1/q}} \nonumber\\
& = &  \left\{\begin{array}{l}
  \frac{ \sqrt{2}(1-r_0)\delta_1}
 {\big(1+ 2^{-1-q}  ((1-r_0)\delta_1)^{2+q}\big)^{1/q}} \\
 \qquad    {\rm if} \ q< 2^{-q} (1-r_0)^{2+q} \delta_1^{2+q},\\
 q^{1/(2+q)} \big(1+\frac{q}{2}\big)^{-1/q} 2^{(1+3q/2)/(2+q)} \\
 \qquad  {\rm if} \ 1\ge q\ge  2^{-q} (1-r_0)^{2+q} \delta_1^{2+q}.
 \end{array}\right.
\end{eqnarray}
Then for any small $q> (e/2+\epsilon)\delta_1^2$ and  sufficiently small $\delta_1>0$, we have that
$q\ge  2^{-q} (1-r_0)^{2+q} \delta_1^{2+q}$ for all $r_0\in (0,1)$. Then applying \eqref{section1.eq9} and
\eqref{maintheorem2.pf.eq2} yields
\begin{eqnarray*}
a(q, \delta_1) & \ge &  \inf_{0<r_0<1} \sup_{\sqrt{2}(1-r_0)\delta_1/2\le y\le 1}
  \frac{2 y}
 {\big(1+ 2^{-q/2}  y^{2+q}\big)^{1/q}}\nonumber\\
 & = &  q^{1/(2+q)} \big(1+\frac{q}{2}\big)^{-1/q} 2^{(1+3q/2)/(2+q)}\nonumber\\
 & \ge  & (1+ \epsilon/e)^{1/2} \delta_1,
 \end{eqnarray*}
where the last inequality holds since
\begin{equation} \label{maintheorem2.pf.eq3}
\lim_{q\to 0} q^{-q/(4+2q)} \big(1+\frac{q}{2}\big)^{-1/q} 2^{(1+3q/2)/(2+q)}=(2/e)^{1/2}.\end{equation}
Therefore
\begin{equation}\label{maintheorem2.pf.eq4}
\limsup_{\delta_1\to 0} \frac{\tilde q_{\max}(\delta_1)}{\delta_1^2}
\le \limsup_{\delta_1\to 0} \frac{(e/2+\epsilon) \delta_1^2 }{\delta_1^2} \le \frac{e}{2}+\epsilon
\end{equation}
for any sufficiently small $\epsilon>0$.

Take $r_0=1-\sqrt{2}/4$ and sufficiently small $\epsilon>0$.
Then  for $q\le  (e/2-\epsilon)\delta_1^2$ and  sufficiently small $\delta_1>0$,
\begin{equation} \label{maintheorem2.pf.eq5}\left\{\begin{array}{l}
(1+r_0 \delta_1) (1+r_0^q \delta_1^q)^{-1/q}\le 2 (3/2)^{-1/q}\le (1-\epsilon/e)^{1/2}\delta_1,\\
\sup_{y\ge 1} y(1+y)^{-1/q}\le  \sup_{y\ge 1}  (1+y)^{1-1/q}\le
 2^{1-1/q}\le (1-\epsilon/e)^{1/2} \delta_1/2,\\
 \sup_{y\ge \sqrt{2}(1-r_0)\delta_1/2}\frac{y}
 {(1+y)^{1/q}} = \frac{\sqrt{2}(1-r_0) \delta_1/2}{
 (1+\sqrt{2}(1-r_0) \delta/2)^{1/q}}\le (1-\epsilon/e)^{1/2} \delta_1/3,
\end{array}\right.
 \end{equation}
 and
 \begin{equation} \label{maintheorem2.pf.eq6}
 \sup_{\sqrt{2}(1-r_0)\delta_1/2\le y\le 1}
  \frac{2 y}
 {\big(1+ 2^{-q/2}  y^{2+q}\big)^{1/q}}\le  (1-\epsilon/e)^{1/2}\delta_1
 \end{equation}
 by \eqref{section1.eq9}, \eqref{maintheorem2.pf.eq2} and \eqref{maintheorem2.pf.eq3}. Therefore
\begin{equation}\label{maintheorem2.pf.eq7}
\liminf_{\delta_1\to 0} \frac{\tilde q_{\max}(\delta_1)}{\delta_1^2}
\ge \limsup_{\delta_1\to 0} \frac{(e/2-\epsilon) \delta_1^2 }{\delta_1^2} \ge \frac{e}{2}-\epsilon
\end{equation}
  by \eqref{maintheorem2.pf.eq5} and \eqref{maintheorem2.pf.eq6}.
 Combining \eqref{maintheorem2.pf.eq4} and \eqref{maintheorem2.pf.eq7} and recalling that $\epsilon>0$
 is a sufficiently small number chosen arbitrarily, we have
  \begin{equation}\label{maintheorem2.pf.eq8}
 \lim_{\delta_1\to 0} \frac{\tilde q_{\max}(\delta_1)}{\delta_1^2}=\frac{e}{2}.
 \end{equation}

 By Corollary \ref{maincorollary}, we have
\begin{equation}
\label{maintheorem2.pf.eq9}
q_{\max}(\delta; m,n,s)\ge \tilde q_{\max}(\sqrt{(1-\delta)/(1+\delta)}).
\end{equation}
This together with \eqref{maintheorem2.pf.eq8} implies that
\begin{equation}
\label{maintheorem2.pf.eq10}
\liminf_{\delta\to 1-} \frac{q_{\max}(\delta; m,n,s)}{1-\delta} \ge
\lim_{\delta\to 1-} \frac{\tilde q_{\max}(\sqrt{(1-\delta)/(1+\delta)})}{(1-\delta)/(1+\delta)}\times \frac{1}{ 1+\delta}=
\frac{e}{4},
\end{equation}
and hence completes the proof.

\bigskip

{\bf Acknowledgement} \quad Part of this work is done when the author is visiting
Vanderbilt University and Ecole Polytechnique Federale de Lausanne   on his sabbatical leave.
The author would like to thank Professors Akram Aldroubi, Douglas Hardin,  Michael Unser and Martin Vetterli for the hospitality and  fruitful discussion, and Professor R. Chartrand for his comments on the early version of this manuscript.

\end{document}